\documentclass[submission,copyright,creativecommons]{eptcs}
\providecommand{\event}{PLACES 2022} 
\usepackage{breakurl}             
\usepackage{underscore}           
\usepackage{mathpartir}
\usepackage{amsmath}
\usepackage{mathtools,amssymb}
\usepackage [utf8] {inputenc}
\usepackage[T1]{fontenc}  
\usepackage{xcolor}
\usepackage{enumerate}
\usepackage{prooftree}

 \newtheorem{theorem}{Theorem}[section]
\newtheorem{definition}[theorem]{Definition}

 \newtheorem{example}[theorem]{Example}
 
\newenvironment{proof}{{\em Proof.}}{\hfill$\square$\medskip}

\newcommand{\refToFigure}[1]{Figure~\ref{#1}}
\newcommand{\refToSection}[1]{Section~\ref{#1}}
\newcommand{\refToExample}[1]{Example~\ref{#1}}

\newcommand{\refToDef}[1]{Definition~\ref{#1}}

\newcommand{\inact}{\ensuremath{\mathbf{0}}}
\newcommand{\coDefGr}{::=_\rho}
\newcommand{\oup}[5]{\bigoplus_{#2\in#3}#1_#2!#4_#2;#5_#2}
\newcommand{\inp}[5]{\Sigma_{#2\in#3}#1_#2?#4_#2;#5_#2}

\newcommand{\PP}{\ensuremath{P}}
\newcommand{\Q}{\ensuremath{Q}}
\newcommand{\R}{\ensuremath{R}}
\newcommand{\pp}{{\sf p}}
\newcommand{\q}{{\sf q}}
\newcommand{\pr}{{\sf r}}
\newcommand{\ps}{{\sf s}}
\newcommand{\pt}{{\sf t}}
\newcommand{\la}{\ell}
\newcommand{\G}{\ensuremath{{\sf G}}}
\newcommand{\End}{\sf{End}}
\newcommand{\agtO}[6]{#1!\{#2_#3.#5_#3;#6_#3\}_{#3\in#4}}
\newcommand{\agtI}[6]{#1?\{#2_#3.#5_#3;#6_#3\}_{#3\in#4}}

\newcommand{\NamedCoRule}[5][]{\IInfer[#1]{#2}{ #3 }{#4}{#5}} 
\newcommand {\IInfer} [5] [] {
  \inferrule*[%
    fraction={===}, 
    left={\textsc{#2}},%
    right={$\begin{array}{l} #5 \end{array}$}, 
    vcenter,%
    #1
  ]%
  {#3}{#4}}
  \newcommand{\rn}[1]{{[\textsc{#1}]}}
  \newcommand{\tyn}[2]{#1\vdash #2}
  
 \newcommand{\pP}[2] {#1[\![\,#2\,]\!]}
 \newcommand{\Nt}{{\mathbb{N}}} 
  \newcommand{\parN}{\mathrel{\|}}
  \newcommand{\plays}[1]{\ensuremath{{\sf Players}(#1)}}
  \newcommand{\set}[1]{\{#1\}}
  
  \newcommand{\Msg}{\mathcal{M}} 
  \newcommand{\mq}[3]{\langle#1,#2,#3\rangle}
  \newcommand{\parG}{\mathrel{\|}}
  \newcommand{\addMsg}[2]{#1\cdot #2}
  \newcommand{\Participants}{\ensuremath{{\sf Part}}}
  \newcommand{\Messages}{\ensuremath{{\sf Lab}}}
  \newcommand{\val}{v}
  \newcommand{\confAs}[2]{#1\parN#2}
  \newcommand{\stackred}[1]{\xrightarrow{#1}}
  \newcommand{\CommAs}[3]{#1!#3.#2}
  \newcommand{\CommAsI}[3]{#1?#3.#2}
  \newcommand{\oi}{input/output}
  \newcommand{\asCom}{\beta}
  \newcommand{\pc}{~|~}
  \newcommand{\comseqA}{ \tau }
   \newcommand{\ee}{\epsilon}
\newcommand{\concat}[2]{\ensuremath{#1\,{\cdot}\,#2}}
\newcommand{\cardin}[1]{\!\!\pc\! #1\!\!\pc\!}
\newcommand{\play}[1]{\ensuremath{{\sf play}(#1)}}
\newcommand{\Seq}[2]{#1;#2}
\newcommand{\NamedRule}[5][]{ \Infer[#1]{#2}{ #3 }{#4}{#5} }
\newcommand {\Infer} [5] [] {
  \inferrule*[%
    left={\textsc{#2}},%
    right={$\begin{array}{l} {#5} \end{array}$}, 
    vcenter,%
    #1
  ]%
  {#3}{#4}}
  \newcommand{\co}{\beta}
   
  \newcommand{\ipth}{\xi}

\newcommand{\IPaths}[1]{{\sf Paths}(#1)}
\newcommand{\weight}{\ensuremath{{\sf depth}}}
\newcommand{\agtoneO}[4]{\Seq{#1!#2.#3}{#4}}
\newcommand{\agtIS}[3]{#1?#2.#3}
\newcommand{\agtSOS}[3]{#1!#3}

\newcommand{\ms}{{\sf m}}

\newcommand{\m}[2]{\ensuremath{{\sf rm}(#1,#2)}}
\newcommand{\n}[2]{{\text{\Large{$\iota$}}(#1,#2)}}

\newcommand{\MM}{\mathbb M}

\title{Asynchronous Sessions with Input Races}

\author{Ilaria Castellani\footnote{This research has been supported by the ANR17-CE25-0014-01 CISC project.} 
\institute{INRIA, Universit\'e C\^ote d'Azur, 
France} 
\email{ilaria.castellani@inria.fr}
\and
Mariangiola Dezani-Ciancaglini
\institute{Dipartimento di Informatica, Universit\`a di Torino, Italy} \email{dezani@di.unito.it}
\and
Paola Giannini\footnote{This original research has the financial support of the Universit\`a  del Piemonte Orientale.}
\institute{DiSSTE,
Universit\`{a} del Piemonte Orientale, Italy} \email{paola.giannini@uniupo.it}
}

\def\titlerunning{Asynchronous Sessions with  Input Races}

\def\authorrunning{I. Castellani, 
M. Dezani-Ciancaglini and P. Giannini}

\begin{document}

\maketitle


\begin{abstract}
  We propose a calculus for asynchronous multiparty sessions where input choices
  with different senders are allowed in processes.  We present a type
  system that accepts such input races provided they do not hinder
  lock-freedom.
\end{abstract}

\section{Introduction}

\label{intro}

The foundational work on multiparty sessions~\cite{CHY08} introduced
the notion of \emph{global type} for specifying the overall behaviour
of multiparty protocols. The criterion for a session implementation to
be correct with respect to its specification was formalised via the
notion of projection: each process implementing the behaviour of a
session participant was required to type-check against the local type
obtained by projecting the global type on that participant.

However, the work~\cite{CHY08} imposed strong restrictions on the
syntax of global types, requiring all initial communications in the
branches of a choice to have the same sender and the same
receiver, 
and every third participant\footnote{We call ``third participant'' any
  participant which is not involved in the first communication of a
  branch.} to have the same behaviour in all branches.  Although these
were useful simplifying assumptions in order to achieve multiparty
session correctness, they limited the expressiveness of global types,
ruling out relevant protocols. For this reason, more permissive choice
constructors were investigated in subsequent
work~\cite{CYH09,CDP12,DY12,LTY15,HY17,CDG18,CDGH19,
  JY20,GHH21,MMSZ21}. A widely adopted relaxation of the choice
operator, originally proposed in \cite{CYH09}, allows third
participants to behave differently in different branches, provided they
are notified of which branch has been chosen. Later
proposals~\cite{DY12,LTY15,HY17} accommodate processes with
output choices among different receivers, for instance a client
choosing among different servers.
On the other hand \emph{input races}, namely input choices among
different senders, continued to be considered as problematic. As a
consequence, common protocols such as a server shared by different
clients could not be specified by means of global types.

Recent proposals introduce more flexibility in input choices for
processes~\cite{CDG18,CDGH19, JY20,GHH21, MMSZ21}. The
work~\cite{GHH21} defines the property of \emph{race-freedom} for
sessions as the absence, at any stage of computation, of a branching
between communications from different senders towards the same
receiver leading to distinct target states. A rather permissive type
system is proposed, which is shown to be both sound and complete for a
range of liveness properties when restricting attention to race-free
sessions.
The work~\cite{JY20} also addresses the input race problem,
referred to as the ``+-problem'' there. While the proposed syntax for
global and local types is completely free, two well-formedness
conditions are imposed on types, which are meant to prevent dangerous
races.  Sessions are synchronous in~\cite{GHH21,JY20} 
and asynchronous in~\cite{MMSZ21}, which is the work
that is closest to ours.
In that paper,  input races are allowed under sophisticated
conditions on projections of global types. These conditions track
causalities between messages, and their soundness proof uses novel
graph-theoretic techniques from the theory of message-sequence charts.

  Consider for example the following session, where two participants
$\pp$ and $\q$ send concurrently a message to a third participant
$\pr$, which is ready to receive both messages in any order:
\begin{example}[Confluent input race]
\label{ex:confluent-race}
$\pP\pp{\pr!\la}\parN\pP\q{\pr!{\la'}}\parN\pP\pr{\pp?\la; \q? \la'
  +\q?\la'; \pp?\la}$
\end{example}\noindent
 No matter whether communication is synchronous or asynchronous,
this session incurs a race\footnote{ Either initially, if communication
  is synchronous, or after performing both outputs,
 if
  communication is asynchronous.}. 
However, this race may be viewed as innocuous since after any
branch is chosen, the input of the other branch is still available,
leading to the same target state. Indeed, the race consists here of a
choice between two different sequentialisations of concurrent
inputs\footnote{If we had a parallel construct $|$ for processes, this
  situation would be represented as $(\pp?\la\, |\, \q?\la')$.}.
This kind of input race will be called \emph{confluent}. 
 The above session  is well typed in~\cite{JY20}, but not in~\cite{GHH21, MMSZ21},
where the syntax of global types forbids different senders in a
choice.  In~\cite{GHH21}, this session is 
also ruled out by the race-freedom condition.

By contrast, the following asynchronous session, where there is an
apparent input race in the process of participant $\pr$, is actually
race-free according to~\cite{GHH21} because it cannot evolve to a
state in which both inputs of $\pr$ are simultaneously enabled.  
This kind of uneffective input race, which results from an agreement
between the senders such that in every computation only one of them
sends a message to the receiver, will be called
\emph{fake}. 
\begin{example}[Fake input race in asynchronous session]
\label{ex:fake-race}
\text{~}\\[3pt]
\centerline{
$\pP\pp{\q!a;\q?a;(\q!\la;\pr!b\oplus\q!\la')}\parN\pP\q{\pp!a;\pp?a;(\pp?\la+\pp?\la';\pr!b)}\parN\pP\pr{\pp?b+\q?b}$
}

\smallskip
This session implements the following protocol 
between Alice, Bob and Carol, represented by
participants $\pp$, $\q$ and $\pr$ respectively:\\
$\bullet$ Alice and Bob send each other the message ``I arrived'' and then
  they read their messages;\\
$\bullet$ Alice sends Bob either the message ``I will tell Carol'',
 after which  
she sends Carol the message ``We arrived'', or the message ``Please
  tell Carol'';\\
$\bullet$ Bob reads either the message ``I will tell Carol'', or the
  message ``Please tell Carol'' 
 after which  he sends Carol the message ``We
  arrived'';\\
$\bullet$ Carol reads the message ``We arrived'' with sender either Alice or Bob.
\end{example}
\noindent
The session of~\refToExample{ex:fake-race} cannot be typed
in~\cite{JY20,GHH21, MMSZ21} because the syntax of global types does
not allow two participants to exchange messages  by first
performing  both outputs and then both inputs. If we omit the
initial exchange of messages between Alice and Bob, 
 the resulting session  can be
typed in~\cite{GHH21, MMSZ21} but not in~\cite{JY20}.  

Our goal is to devise a type system for asynchronous sessions that is
permissive enough to accept the sessions of
\refToExample{ex:confluent-race} and \refToExample{ex:fake-race},
while rejecting dangerous races that could lead to deadlock or
starvation.   In particular, we will not be able to type the
sessions discussed in the introductions of~\cite{JY20,MMSZ21}, since
they both have a possibility of starvation.

For typing asynchronous sessions we 
use global types that
split communications into outputs and inputs, following the approach
advocated in~\cite{CDG21,DGD22}.
For instance, the session of \refToExample{ex:fake-race}  has 
the following global type: 
$
\pp!\q.a;\q!\pp.a;\pp?\q.a;\q?\pp.a;\pp!\{\q.\la;\G_1\,,\,
  \q.\la’; \G_2\}$ where $\G_1=\pp!\pr.b;\q?\pp.\la;\pr?\set{\pp.b\,,\,\q.b}$ and $\G_2=\q?\pp.\la’;\q!\pr.b;\pr?\set{\pp.b\,,\,\q.b}
  $. 
  Here
$\pp!\q.a$ denotes a send from $\pp$ to $\q$ of label $a$,
$\pp?\q.a$ denotes a receive by $\pp$ from $\q$ of label $a$, $\pp!\set{\q.\la
  ; \G_1\,,\, \q.\la’; \G_2}$ is an output choice with sender
$\pp$ and receiver $\q$, and $\pr?\set{\pp.b\,,\,\q.b}$ is an
input choice with receiver $\pr$ and senders $\pp$ and $\q$.

The rest of the paper is organised as follows. In
\refToSection{sec:as} we introduce our calculus for asynchronous
sessions. In \refToSection{sec:types} we define the syntax and
semantics of global types. In \refToSection{sec:ts} we present our
type system,  illustrate it with some examples,  and
establish the properties of Subject Reduction, Session Fidelity and
Lock-freedom. We conclude in \refToSection{sec:relw} with a discussion
on future and related work.

\section{Asynchronous Sessions}\label{sec:as}

We assume the following base sets: \emph{participants}, ranged
over by $\pp,\q,\pr$ and forming the set $\Participants$, and
\emph{labels}, ranged over by $\la,\la',\dots$ and forming the set
 $\Messages$.

 \begin{definition}[Processes]\label{p} 
  {\em Processes} are defined by:\\
\centerline{$\begin{array}{rcl}
\PP & \coDefGr  & 
\inact
\mid
\oup\pp{i}{I}{\la}{\PP}%
\mid
\inp\pp{i}{I}{\la}{\PP}%
\end{array}
$}
where $I\neq\emptyset$ and $\pp_h!\la_h\neq\pp_k!\la_k$ and
$\pp_h?\la_h\neq\pp_k?\la_k$  for  $h, k\in I$ and $h\neq k$.  
\end{definition}
 A process may be terminated, or it is an internal choice of
outputs or an external choice of inputs.  The symbol $ \coDefGr$,
in \refToDef{p} and in later definitions, indicates that the
productions should be interpreted \emph{coinductively} (they define
possibly infinite processes) and that we focus on \emph{regular}
terms, namely, terms with finitely many distinct subterms.  In this
way, we only obtain processes which are solutions of finite sets of
equations, see~\cite{Cour83}. So, when writing processes, we shall use
(mutually) recursive equations.

In the following, we will omit trailing $\inact$'s when writing
processes.  

In a full-fledged calculus,  processes would exchange
labels of the form $\la(\val)$, where $\val$ is a value.
For simplicity, we consider only pure labels here.

In our calculus,  asynchronous communication is handled in the standard
way, by storing sent labels in a queue together with sender and
receiver names. Receivers may then fetch messages from the queue
when required.
 We define \emph{messages} to be
triples $\mq\pp{\la}\q$,
where $\pp$ is the sender and $\q$ is the receiver, 
and \emph{message queues} (or simply \emph{queues}) to be possibly empty
sequences of messages: \\ 
\centerline{$\Msg::=\emptyset \mid \addMsg{\mq\pp{\la}\q}{\Msg} $}

The order of messages in the queue is the order in which they will be
read. Since the only reading order that matters is that between
messages with the same sender and the same receiver, we consider
message queues modulo the structural equivalence given by:\\
\centerline{$
\addMsg{\addMsg{\Msg}{\mq\pp{\la}\q}}{\addMsg{\mq\pr{\la'}\ps}{\Msg'}}\equiv
  \addMsg{\addMsg{\Msg}{\mq\pr{\la'}\ps}}{\addMsg{\mq\pp{\la}\q}{\Msg'}}
  ~~\text{if}~~\pp\not=\pr~~\text{or}~~\q\not=\ps 
$}

Sessions  are composed by a number of located processes of the form
$\pP{\pp}{\PP}$, each enclosed within a different
participant $\pp$, and by a message queue. 

\begin{definition}[Networks and Sessions] 
{\em  Networks} are defined by:\\
\centerline{$ \Nt = \pP{\pp_1}{\PP_1} \parN
\cdots \parN \pP{\pp_n}{\PP_n}$ with
 $\pp_h \neq \pp_k $ for any $h \neq k$}
{\em  Sessions} are defined by:\\
\centerline{$
\Nt \parallel \Msg$} where $\Nt$
   is a network,  and $\Msg$ is a message queue.
\end{definition}

We assume the standard structural congruence  $\equiv$  on
networks\footnote{ By abuse of notation, we use the same symbol as for
structural equivalence on queues.}, stating that
parallel composition is associative and commutative and has neutral
element $\pP\pp\inact$ for any fresh $\pp$.

If $\PP\neq\inact$ we write $\pP{\pp}{\PP}\in\Nt$ as short for
$\Nt\equiv\pP{\pp}{\PP}\parN\Nt'$ for some $\Nt'$.  This abbreviation
is justified by the associativity and commutativity of parallel composition.

To define the {\em operational semantics} of sessions, we use an LTS
whose transitions are decorated by outputs or inputs.  Therefore, we
define the set of {\em \oi\ communications} (
communications for short), ranged over by $\asCom$, $\asCom'$, to be
$\{ \CommAs{\pp}{\la}{\q}, \CommAsI{\pp}{\la}{\q} \pc
\pp,\q\in \Participants, \la \in \Messages\}$, where
$\CommAs{\pp}{\la}{\q}$ represents the  output  of  the  label $\la$ from
participant $\pp$ to participant $\q$, and $\CommAsI{\pp}{\la}{\q}$
the  input  by participant $\pp$ of the label $\la$ sent by
participant $\q$. 

The LTS semantics of networks,  defined modulo $\equiv$, is
specified by  the two Rules $\rn{Send}$ and $\rn{Rcv}$ given in
\refToFigure{fig:asynprocLTS}.  Rule $\rn{Send}$ allows a participant
$\pp$ with an internal choice (a sender) to send to the participant
$\q_k$ the label $\la_k$ by adding it to the queue. Symmetrically,
Rule $\rn{Rcv}$ allows a participant $\q$ with an external choice (a
receiver) to read the label $\la_k$ sent by participant $\pp_k$,
provided this label is among the $\la_j$'s specified in the choice.
Thanks to structural equivalence, the first 
message from $\pp_k$  to $\q$ that appears in the queue, if any, 
can always be moved to the top of the queue. 

 A key role in this paper  is played by (possibly empty) sequences of communications.
As usual we define them as traces.
\begin{definition}[Traces]\label{tra}  (Finite) traces are defined by:\\
\centerline{$\comseqA::=\ee\mid\concat\beta\comseqA$}
 We use $\cardin{\comseqA}$ to denote the length of the trace $\comseqA$.
\end{definition}
\noindent
When $\comseqA=\concat{\beta_1}{\concat\ldots{\beta_n}}$ ($n\geq 1)$
we write $\Nt\parN\Msg\stackred{\comseqA}\Nt'\parN\Msg'$ as short for\\
\centerline{$\Nt\parN\Msg\stackred{\beta_1}\Nt_1\parN\Msg_1\cdots\stackred{\beta_n}\Nt_{n}\parN\Msg_{n} 
 =  \Nt'\parN\Msg'$}
 
 \bigskip

We now introduce the notion of player,  which is characteristic of
asynchronous communication, where only one of the involved
participants is active, namely the sender for an output communication
and the receiver for an input communication. 
 The  player of a communication $\beta$ is  the 
participant who is active in $\beta$.  The set of players of a
trace is then obtained by collecting the players of all its communications. 

\begin{definition}[Players of
   communications and traces]
  \label{def:partAct}
  We denote by \play{\beta} the
  {\em player of a communication  
    $\beta$}
defined by\\
\centerline{
 $\play{\CommAs{\pp}{\la}{\q}}=
    \play{\CommAsI{\pp}{\la}{\q}}=\pp$ }
We denote by \plays{\comseqA} the
  {\em set of players of a trace  
    $\comseqA$}
defined by\\
\centerline{$\plays\ee=\emptyset\qquad\plays{\concat\beta\comseqA}=
 \set{\play\beta}\cup\plays\comseqA$}
 \end{definition}
 
 \begin{figure}
 \centerline{$
\begin{array}{c}  
\\[5pt]
\confAs{\pP{\pp}{\oup\q{i}{I}{\la}{\PP}}\parN\Nt}{\Msg} \stackred{\CommAs\pp{\la_k}{\q_k}}
  \confAs{\pP{\pp}{\PP_k}\parN\Nt}{\addMsg{\Msg}{\mq\pp{\la_k}{\q_k}}}\quad \text{ where  }\ 
   k \in I  \quad
   {~~~~~~\rn{Send}}
   \\[3pt]
\confAs{\pP{\q}{\inp\pp{j}{J}{\la}{\Q}}\parN\Nt}{\addMsg{\mq{\pp_k}{\la_k}\q}{\Msg}}\stackred{\CommAsI\q{\la_k}{\pp_k}}
 \confAs{\pP{\q}{\Q_k}\parN\Nt}{\Msg}\quad  \text{ where  }\ 
 k \in J \quad
  {~~~~~~\rn{Rcv}}
   \\[3pt]
\end{array}
$} 
\caption{LTS for sessions.}\label{fig:netredAs}\label{fig:asynprocLTS}
\end{figure}

\section{Global Types}\label{sec:types}

As in~\cite{CDG21,DGD22}, our global types can be obtained  from
the standard ones~\cite{CHY08,CHY16} by splitting output and
input communications. The novelty is that we 
allow multiple receivers in
output choices and multiple senders in input choices.

\begin{definition}[Global  types] 
\label{gt}
\emph{Global types} 
$\G$ are defined by the following grammar:\\ 
\centerline{$\begin{array}{rcl}
\G & \coDefGr &    \agtO{\pp}{\q}i I{\la}{\G}
              \mid \agtI \pp\q i I \la \G   
              \mid \End
\end{array}$}
where $I\neq\emptyset$, $\pp\neq \q_i$ for all $i\in I$ and
$\q_h.\la_h\not=\q_k.\la_k\,$ 
for $h, k\in I$ and $h\neq k$. 
\end{definition}

As for processes, $ \coDefGr$ indicates that global types are
coinductively defined and \emph{regular}.

The global type $\agtO{\pp}{\q}i I{\la}{\G}$ specifies that player
$\pp$ sends the label $\la_k$ with $k\in I$ to participant $\q_k$ and
then the interaction described by the global type $\G_k$ takes place.
The global type $\agtI \pp\q i I \la \G $ specifies that player $\pp$
receives the label $\la_k$ with $k\in I$ from participant $\q_k$ and
then the interaction described by the global type $\G_k$ takes place.

We
define $\plays{\G}$ as the smallest set satisfying the following equations:\\
\centerline{$\plays{ \End } = \emptyset \quad 
\plays{ \agtO{\pp}{\q}i I{\la}{\G}} = \plays{ \agtI \pp\q iI\la \G } =
\set{\pp}\cup\bigcup_{i\in I}\plays{\G_i} $} 
The regularity of global types ensures that the sets of players are finite.
In~\refToSection{sec:as}
we used the same notation for the players of traces. In all cases, the
context should make it easy to understand which function is in use.
 
To avoid starvation we require global types to satisfy a boundedness condition. 
To formalise boundedness we use $\ipth$ to denote a \emph{path} in
global type trees, i.e., a possibly infinite sequence of
communications $\co$.  Note that a finite path is a trace in the
sense of \refToDef{tra}. 
 We extend the notation $\cdot$ to denote 
also  the concatenation of a finite sequence with a possibly
infinite sequence.  
The function ${\sf Paths}$ gives the set of
\emph{paths} of a global type, which is the greatest set such that:\\
\centerline{$\begin{array}{ll} 
\IPaths{\End} &= \set{\epsilon}  \\ 
\IPaths{\agtO{\pp}{\q}i I{\la}{\G}} &= \bigcup_{i\in I} \set{ \concat{\CommAs\pp{\la_i}\q}{\ipth} \mid \ipth\in\IPaths{\G_i} } \\
\IPaths{\agtI \pp\q i I \la \G } &= \bigcup_{i \in I} \set{ \concat{\CommAsI\pp{\la_i}\q}{\ipth} \mid \ipth\in \IPaths{\G_i} } 
\end{array}$} 
\smallskip

 If $x \in \mathbf{N} \cup \set{\omega}$ is the
length of $\ipth$, we denote by $\ipth[n]$ the $n$-th communication in the path $\ipth$,
where $1\le n < x$. 
It is handy to define the \emph{depth} of a player $\pp$ in a global type 
$\G$, $\weight(\G,\pp)$. 

\begin{definition}[Depth of a  player]\label{def:depth}
Let $\G$ be a global type. 
For ${\ipth\in\IPaths{\G}}$ set\\ \centerline{$\weight(\ipth,\pp) = \inf \{ n \mid \play{\ipth[n]} = \pp \}$}
and define $\weight(\G,\pp)$, the \emph{depth} of $\pp$ in $\G$, as follows:\\ 
\centerline{$\weight(\G,\pp) = \begin{cases} 
\sup \{ \weight(\ipth,\pp) \mid \ipth \in \IPaths{\G} \}& \pp \in \plays{\G} \\ 
0 & \text{otherwise} 
\end{cases}
$}
\end{definition}
Note that $\weight(\G,\pp)=0$ iff  $\pp \not\in \plays{\G}$. 
 Moreover, if  $\pp\ne\play{\ipth[n]}$  for 
all $n\in\mathbf{N}$,  then $\weight(\ipth,\pp) = \inf\, \emptyset = \infty$. 
Hence, if $\pp$ is a player of a global type $\G$ 
and there is some path in $\G$ where $\pp$ does not occur as a player, 
then $\weight(\G,\pp) = \infty$.

\begin{definition}[Boundedness]\label{def:bound}
A global type $\G$ is \emph{bounded} if  $\weight(\G',\pp)$ is finite
for all participants $\pp\in\plays{\G}$ and  all types  
$\G'$ which occur in   $\G$. 
\end{definition}
\begin{example}\label{b} The following example shows the  necessity  of considering all types  occurring in a global type for defining boundedness. 
Consider 
 $\G= \agtoneO \pr\q {\la}{\agtIS\q\pr  {\la};\G'} $,  where\\
\centerline{$\G'=\agtSOS \pp\q {\{\Seq{\q.\la_1}{\agtIS \q \pp \la_1;\agtoneO \q\pr {\la_3}\agtIS\pr\q {\la_3}}\, ,\,\Seq{\q.\la_2}{\Seq{\agtIS \q\pp \la_2}{\G'}}\}}
$} Then we have: $
 \weight(\G,\pp)=3,\weight(\G,\q)=2 ,\weight(\G,\pr)=1
$, 
whereas 
$\weight(\G',\pp)=1,\weight(\G',\q)=2,\weight(\G',\pr)=\infty
$.
\end{example}
Since global types are regular the boundedness condition is decidable.

\begin{figure}
\begin{math} 
\begin{array}{c}
\NamedRule{{\rn{Top-Out}}}
 { } 
 {\agtO{\pp}{\q}i I{\la}{\G}\parG\Msg \stackred{\CommAs\pp{\la_h}{\q_h}}\G_h\parG\addMsg\Msg{\mq\pp{\la_h}{\q_h}}}
 {h\in I}
\\[4ex]
\NamedRule{{\rn{Top-In}}}
 { } 
 {\agtI \pp\q i I\la \G\parG\addMsg{\mq{\q_h}{\la_h}\pp}\Msg \stackred{\CommAsI{\q_h}{\la_h}\pp}\G_h\parG\Msg}
 {h\in I}
\\[4ex]
\NamedRule{{\rn{Inside-Out}}}
 {\G_i\parG \Msg\cdot\mq\pp{\la_i}{\q_i} \stackred\asCom\G'_i \parG \Msg'\cdot\mq\pp{\la_i}{\q_i} \quad \forall i \in I} 
 {\agtO{\pp}{\q}i I{\la}{\G}\parG\Msg \stackred \asCom\agtO{\pp}{\q}i I{\la}{\G'}\parG\Msg' }  
 {\pp\ne\play{\asCom} }
\\[4ex]
\NamedRule{{\rn{Inside-In}}}
 {\G_{ j  }\parG \Msg\stackred\asCom\G_{ j  }'\parG \Msg'\quad \forall { j } \in J}
 {\agtI \pp\q i I \la \G\parG \Msg \stackred\asCom
 \agtI \pp\q i {I} \la {\G'}\parG  {\Msg'} 
}  
 {\begin{array}{c} J= 
\m{\set{\mq{\q_i}{\la_i}\pp}_{i\in I}}\Msg\ne\emptyset\\
 \pp\ne\play{\asCom} \quad  
     \beta \neq \q_l!\pp. {\la_l} \\
\G'_l = \G_k ~~ k\in J~~\forall l\in I\backslash J  
 \end{array}}
\end{array}
\end{math} 
\caption{LTS for type configurations. }\label{fig:ltsgtAs}
\end{figure}

Global types in parallel with queues, dubbed \emph{type
  configurations}, are given semantics by means of the LTS in
\refToFigure{fig:ltsgtAs}.  The first two rules allow top level
outputs and inputs to be performed in the standard way.  The remaining
two rules allow communications to be performed inside output and input
choices.  These inside rules are needed to enable interleaving between
independent communications despite the sequential structure of global
types.  For example, we want to allow
$\Seq{\CommAs\pp{\la}\q}{\CommAs\pr{\la'}\ps}\parG\emptyset\stackred{\CommAs\pr{\la'}\ps}{\CommAs\pp{\la}\q}\parG\mq\pr{\la'}\ps$
when $\pp\neq\pr$, because, intuitively, outputs performed by
different players should be independent.  This justifies the condition
$\pp\ne\play{\asCom}$ in Rules \rn{Inside-Out} and \rn{Inside-In}. In
Rule \rn{Inside-Out} we require all branches to be able to perform
the $\co$ transition. This avoids for example:\\
\centerline{$\begin{array}{c}\pp!\set{\q.\la;\q?\pp.\la;\pr!\pp.\la;\pp?\pr.\la\,,\,\q.\la';\q?\pp.\la';\pr!\pp.\la';\pp?\pr.\la'} \parN\emptyset\stackred{\pr!\pp.\la}\\
    \pp!\set{\q.\la;\q?\pp.\la  ; \pp?\pr.\la\,,\,  \q.\la';\q?\pp.\la';\pr!\pp.\la';\pp?\pr.\la'} \parN\mq\pr\la\pp\end{array}$} \\
which, in case we choose the right branch, leads to the configuration
$\pp?\pr.\la' \parN\addMsg{\mq\pr\la\pp}{\mq\pr{\la'}\pp}$.

The shapes of the queues  appearing in the premise of  Rule \rn{Inside-Out}  ensure  that
$\co$ is not the matching input for any output in the choice. In Rule
\rn{Inside-In}, we consider only the branches with corresponding
messages on top of the queue  (called {\em live} branches),  using the index set of ready
messages  
$\m{\set{\mq{\q_i}{\la_i}\pp}_{i\in I}}\Msg$  
defined as follows, where $\ms$ ranges over messages. 
\begin{definition}\label{rm} Given a set of messages 
  $\set{\ms_i}_{i\in I}$ and a queue $\Msg$,
  the index set of the ``ready messages''  in this set  is defined by: 
  $\m{\set{\ms_i}_{i\in I}}\Msg
=\set{i\in
      I\mid\Msg\equiv\addMsg{\ms_i}{\Msg_i}}$.
 \end{definition}
 \noindent
 The mapping
 ${\sf rm}$ plays a crucial role also in the typing rule for input
 choices, as we will see in \refToSection{sec:ts}.  
  The condition $J\neq\emptyset$ means that there is at least one live branch. 
 The
 condition $\beta \neq \q_l!\pp. {\la_l}$  for all  $l\in I\backslash J$
 ensures that the occurrence of $\beta$ does not 
generate a message that would  ``awaken''  some dead,  i.e. not live,   branch of the choice.
 In the resulting choice,
 the dead branches become an arbitrary live branch (condition 
 $\G'_l = \G_k$ for some $k\in J$ and for all  $l\in I\backslash J$). 
In fact, such branches could also be omitted as they can never be awaken.

 It is easy to check that the LTS of type configurations preserves
 boundedness of global types.  Therefore, from now on we will
assume all our global types to be bounded.

\section{Type System}\label{sec:ts}

\begin{figure}[h]
\begin{math}
\begin{array}{c}
\NamedCoRule{\rn{End}}{}{ \tyn\End{\pP\pp\inact\parN\emptyset} }{} \\[6ex] 
\NamedCoRule{\rn{Out}}{
  \tyn{\G_i}{\pP\pp{\PP_i}\parN\Nt\parN\Msg\cdot\mq\pp{\la_i}{\q_i}}\\ \plays{\G_i}\setminus\set\pp=\plays\Nt\ \ \forall i \in I
}{ \tyn{\agtO{\pp}{\q}{i}{I}{\la}{\G}}{\pP\pp{\oup{\q}{i}{I}{\la}{\PP}}\parN\Nt\parN\Msg}}{} 
\\[6ex] 
\NamedCoRule{\rn{In}}{
  \tyn{\G_{j}}{\pP\pp{\PP_{j}}\parN\Nt\parN\Msg_{j}}\quad\n{\G_j}{\MM}\quad \forall  j  \in J
\\
  \plays{\G_{ i  }}\setminus\set\pp=\plays\Nt\quad \forall  i  \in I 
}{ \tyn{\agtI{\pp}{\q}i I{\la}{\G}}{\pP\pp{\inp{\q}{h}{H}{\la}{\PP}}\parN\Nt\parN\Msg} }
{\begin{array}{l}
J=\m{\set{\mq{\q_i}{\la_i}\pp}_{i\in I}}\Msg \\
~~=\m{\set{\mq{\q_h}{\la_h}\pp}_{h\in H}}\Msg\neq\emptyset\\
  \Msg\equiv\addMsg{\mq{\q_{j}}{\la_{j }}\pp}{\Msg_{j}}~\forall j\in J\\
 \MM=\set{{\mq{\q_{ l }}{\la_{ l  }}\pp}\mid l \in (I\cup H) \setminus J\  \&\ \q_l\neq\q_j\ \forall j\in J} \end{array}}
\end{array} 
\end{math}
\caption{Typing rules for sessions.}\label{fig:cntr} 
\end{figure}

Global types are an abstraction of sessions. Usually,
global types are projected to participants, yielding local types which
are assigned to processes. The simplicity of our
calculus and the flexibility of our global types allow us to formulate
a type system where global types are directly derived for sessions,
using  judgements of the form $\tyn\G{\Nt\parN\Msg}$. The typing rules
are given in~\refToFigure{fig:cntr}. 

Rules \rn{Out} and \rn{In} just add simultaneously outputs and inputs
to global types and to  the corresponding processes inside
networks.  The condition $ \plays{\G_i}\setminus\set\pp=\plays\Nt$ for
all $i \in I$ ensures that all players in $\Nt$ are also players in
$\G$.  For example, this condition prevents the derivation of 
$\tyn\G{\pP\pp\PP\parN\pP\q\Q}$ with $\G=\Seq{\CommAs\pp{\la}\q}\G$
and $\PP=\Seq{\q!\la}\PP$ and $\Q$ arbitrary. 

Rule \rn{Out} considers all branches of the global type, since the
choice of the sent message is arbitrary. This rule requires that the
session 
resulting from the output of a branch be typed with the corresponding
branch of the global type.

Rule \rn{In} requires that the global type and the process read the
same messages on the queue.  To this end, it uses the index set of
ready messages defined in \refToDef{rm}, collecting the indices of the
live branches of the global type and of the input process\footnote{ As
  for global types, a branch of the input process is live if it has a
  corresponding message on top of the queue, and dead otherwise.}, and
asking them to be equal (condition 
$\m{\set{\mq{\q_i}{\la_i}\pp}_{i\in I}}\Msg=\m{\set{\mq{\q_h}{\la_h}\pp}_{h\in H}}\Msg$.  This
set of indices must not be empty (condition $J\neq\emptyset$).  Only
the branches of the global type and of the input process thus selected
are compared 
in the premises of Rule \rn{In}. Note that in this way
we allow more freedom than in the synchronous subtyping for session types~\cite{DH12}. 
In Rule \rn{In},   in order to ensure 
the condition $\beta \neq
\q_l!\pp. {\la_l}$ for all $l\in I\backslash J$ required by the 
 transition  Rule \rn{Inside-In},  we want 
 to prevent the enqueuing of 
messages 
that would transform a dead branch of the process or of the global
type into a live branch.
 To this end, we introduce a predicate  which forbids a global type to generate such
messages.  Let  $\MM$ range over sets of messages.
\begin{definition} The type $\G$ is {\em inactive} for the set of messages $\MM$, if $\n\G\MM$ holds, where:\\ 
\centerline{
$\begin{array}{c}
\n\End\MM
\quad\quad\quad\n{\agtI{\pp}{\q}{i}{I}{\la}{\G}}\MM\quad\text{if}\ \n{\G_i}{\MM}~~\forall  i  \in I\\
\n{\agtO{\pp}{\q}{i}{I}{\la}{\G}}\MM\quad\text{if}\ {\mq\pp{\la_i}{\q_i}\not\in\MM} \ \text{and}\ \n{\G_i}{\MM}~~\forall  i  \in I
 \end{array}
$}
\end{definition}
\noindent The predicate $\n\G\MM$ looks for outputs in $\G$ which
produce messages in $\MM$. The regularity of global types 
 guarantees   
the
computability of this predicate.  Notice that $\n\G\MM$
also ensures that the network cannot produce messages in $\MM$. This is due to
the typing Rule \rn{Out} prescribing that messages put on the queue by the global type be the same as the ones of the network.

For example consider the following   sequence of transitions \\
\centerline{$\begin{array}{lll}\pP\q{\pr!\la'}\parN\pP\pr{\pp?\la;\q?\la'+\q?\la';\pp?\la'}\parN\mq\pp\la\pr&\stackred{\CommAs\q{\la'}{\pr}}&\pP\pr{\pp?\la;\q?\la'+\q?\la';\pp?\la'}\parN\mq\pp\la\pr\cdot\mq\q{\la'}\pr\\
    &\stackred{\CommAsI\pr{\la'}{\q}}&\pP\pr{\pp?\la'}\parN\mq\pp\la\pr\end{array}$}
Since the input and the message in $\pP\pr{\pp?\la'}\parN\mq\pp\la\pr$
do not match, this session cannot be typed and therefore also the
session
$\pP\q{\pr!\la'}\parN\pP\pr{\pp?\la;\q?\la'+\q?\la';\pp?\la'}\parN\mq\pp\la\pr$
should not be typable.
Without checking the inactivity predicate we can type this session by the global type\\
\centerline{$(\ast)\quad\pr?\set{\pp.\la;\q!\pr.\la';\pr?\q.\la',\q.\la';\q!\pr.\la';\pr?\pp.\la'}$}
as follows:\\
\centerline{$\prooftree \prooftree \prooftree
  \tyn{\End}{\pP\pr{\inact}\parN\emptyset} \Justifies
  \tyn{\pr?\q.\la'}{\pP\pr{\q?\la'}\parN\mq\q{\la'}\pr}
\endprooftree
\Justifies
\tyn{\q!\pr.\la';\pr?\q.\la'}{\pP\q{\pr!\la'}\parN\pP\pr{\q?\la'}\parN\emptyset}
\endprooftree
\Justifies
\tyn{\pr?\set{\pp.\la;\q!\pr.\la';\pr?\q.\la',\q.\la';\q!\pr.\la';\pr?\pp.\la'}}{\pP\q{\pr!\la'}\parN\pP\pr{\pp?\la;\q?\la'+\q?\la';\pp?\la'}\parN\mq\pp\la\pr}
\endprooftree
$}
\smallskip

\noindent
The problem here is that Rule \rn{In} does not check the
 dead branches of the global type. 
The role of the inactivity predicate is just to ensure that the
transitions will not awake dead branches.
This is done by checking the outputs in the 
 live branches. 
 In  this  example the output $\q!\pr.\la'$ is in the branch
 starting with the input $\pr?\pp.\la$ and the queue contains
 $\mq\pp\la\pr$.  So the typing Rule \rn{In} cannot be applied since
 $\n{\q!\pr.\la';\pr?\q.\la'}{\set{\mq\q{\la'}\pr}}$ does not hold.

Notice that the session in \refToExample{ex:confluent-race} has the transition\\
\centerline{$\begin{array}{lll}\pP\pp{\pr!\la}\parN\pP\q{\pr!{\la'}}\parN\pP\pr{\pp?\la;
      \q? \la'+\q?\la'; \pp?\la}\parN \emptyset 
    &\stackred{\CommAs\pp{\la}{\pr}}&\pP\q{\pr!{\la'}}\parN\pP\pr{\pp?\la;
      \q? \la'+\q?\la'; \pp?\la}\parN\mq\pp\la\pr\end{array}$} 
      
      \smallskip \noindent
and the resulting session differs from that of the previous example
only for the label of the last input.  Correspondingly, the global
types 
$\pr?\set{\pp.\la;\q!\pr.\la';\pr?\q.\la'\,,\,\q.\la';\underline{\pr?\pp.\la}}$
and
$\pr?\set{\pp.\la;\q!\pr.\la';\pr?\q.\la'\,,\,\q.\la';\underline{\pr?\pp.\la'}}$
only differ for the labels of the underlined inputs.  Therefore
$\pr?\set{\pp.\la;\q!\pr.\la';\pr?\q.\la'\,,\,\q.\la';{\pr?\pp.\la}}$ cannot be derived for the session
$\pP\q{\pr!{\la'}}\parN\pP\pr{\pp?\la; \q? \la'+\q?\la';
  \pp?\la}\parN\mq\pp\la\pr$, since as we just saw 
the predicate
$\n{\q!\pr.\la';\pr?\q.\la'}{\set{\mq\q{\la'}\pr}}$ does not
hold. 
In fact, this is expected since the session can do a transition
$\stackred{\CommAs\q{\la'}{\pr}}$ that the type configuration cannot mimic. 
 On
the other hand, this session can be typed 
by the global type\\ \centerline{$\q!\pr.\la';\pr?\set{\pp.\la;\pr?\q.\la'\,,\,\q.\la';\pr?\pp.\la}$}
as follows:\\
  \centerline{$\prooftree
  \prooftree
  \prooftree \tyn{\End}{\pP\pr{\inact}\parN\emptyset}  \Justifies \tyn{\pr?\q.\la'}{\pP\pr{\q? \la'}\parN\mq\q{\la'}\pr}\endprooftree
  \qquad  \prooftree \tyn{\End}{\pP\pr{\inact}\parN\emptyset}  \Justifies \tyn{\pr?\pp.\la}{\pP\pr{\pp?\la}\parN\mq\pp\la\pr}\endprooftree
  \Justifies
  \tyn{\pr?\set{\pp.\la;\pr?\q.\la'\,,\,\q.\la';\pr?\pp.\la}}{\pP\pr{\pp?\la; \q? \la'+\q?\la'; \pp?\la}\parN\mq\pp\la\pr\cdot\mq\q{\la'}\pr}
  \endprooftree
  \Justifies
  \tyn{\q!\pr.\la';\pr?\set{\pp.\la;\pr?\q.\la'\,,\,\q.\la';\pr?\pp.\la}}{\pP\q{\pr!{\la'}}\parN\pP\pr{\pp?\la; \q? \la'+\q?\la'; \pp?\la}\parN\mq\pp\la\pr}
\endprooftree
$}  
\smallskip

\noindent
In this derivation  Rule \rn{Out} is applied first, and thus Rule
\rn{In} is applied only when the queue contains the matching messages for
both branches. Therefore Rule \rn{In} checks the continuations of both branches, and 
the inactivity predicate holds trivially for each of them.
%
%

Notice that 
 bringing forward 
the output from $\q$ to $\pr$ in the global
type  $(\ast)$ 
does not   enable us to type:  
\\ \centerline{$\pP\q{\pr!\la'}\parN\pP\pr{\pp?\la;\q?\la'+\q?\la';\pp?\la'}\parN\mq\pp\la\pr$} 
In fact,  we cannot complete  the derivation:\\[3pt]
  \centerline{$\prooftree
  \prooftree
  \prooftree \tyn{\End}{\pP\pr{\inact}\parN\emptyset}  \Justifies \tyn{\pr?\q.\la'}{\pP\pr{\q? \la'}\parN\mq\q{\la'}\pr}\endprooftree
  \qquad  \tyn{\pr?\pp.\la'}{\pP\pr{\pp?\la'}\parN\mq\pp\la\pr}
  \Justifies
  \tyn{\pr?\set{\pp.\la;\pr?\q.\la'\,,\,\q.\la';\pr?\pp.\la'}}{\pP\pr{\pp?\la; \q? \la'+\q?\la'; \pp?\la'}\parN\mq\pp\la\pr\cdot\mq\q{\la'}\pr}
  \endprooftree
  \Justifies
  \tyn{\q!\pr.\la';\pr?\set{\pp.\la;\pr?\q.\la'\,,\,\q.\la';\pr?\pp.\la'}}{\pP\q{\pr!{\la'}}\parN\pP\pr{\pp?\la; \q? \la'+\q?\la'; \pp?\la'}\parN\mq\pp\la\pr}
\endprooftree
$} 

\smallskip
\noindent
 Indeed,  we cannot apply Rule \rn{In}  to derive the top
right judgement 
$\tyn{\pr?\pp.\la'}{\pP\pr{\pp?\la'}\parN\mq\pp\la\pr}$, since the
only input does not match the message in the queue.  

We can also type the following recursive version of \refToExample{ex:confluent-race}:\\
\centerline{$\pP\pp{\PP}\parN\pP\q{\Q}\parN\pP\pr{\R}$}
where $\PP=\pr!\la  ; \PP$, $\Q=\pr!{\la'};\Q$ and $\R=\pp?\la; \q? \la';\R+\q?\la'; \pp?\la;\R$. A suitable global type is\\ 
\centerline{$\G=\pp!\pr.\la;\q!\pr.\la';\pr?\set{\pp.\la;\pr?\q.\la';\G\,,\,\q.\la';\pr?\pp.\la;\G}$}
as shown  by  the following derivation:\\
  \centerline{$\prooftree
  \prooftree
  \prooftree
  \prooftree 
  \prooftree
       \vdots 
  \Justifies
  \tyn{\G}{\pP\pp{\PP}\parN\pP\q{\Q}\parN\pP\pr{\R}\parN\emptyset}
   \endprooftree  
  \Justifies \tyn{\pr?\q.\la';\G}{\pP\pp{\PP}\parN\pP\q{\Q}\parN\pP\pr{\q? \la';\R}\parN\mq\q{\la'}\pr}\endprooftree
  \qquad  \prooftree 
     \prooftree
       \vdots 
  \Justifies
    \tyn{\G}{\pP\pp{\PP}\parN\pP\q{\Q}\parN\pP\pr{\R}\parN\emptyset} 
    \endprooftree 
    \Justifies \tyn{\pr?\pp.\la;\G}{\pP\pp{\PP}\parN\pP\q{\Q}\parN\pP\pr{\pp?\la;\R}\parN\mq\pp\la\pr}
  \endprooftree
  \Justifies
  \tyn{\pr?\set{\pp.\la;\pr?\q.\la';\G\,,\,\q.\la';\pr?\pp.\la;\G}}{\pP\pp{\PP}\parN\pP\q{\Q}\parN\pP\pr{\R}\parN\mq\pp\la\pr\cdot\mq\q{\la'}\pr}
  \endprooftree
  \Justifies
  \tyn{\q!\pr.\la';\pr?\set{\pp.\la;\pr?\q.\la';\G\,,\,\q.\la';\pr?\pp.\la;\G}}{\pP\pp{\PP}\parN\pP\q{\Q}\parN\pP\pr{\R}\parN\mq\pp\la\pr}
\endprooftree
 \Justifies
  \tyn{\G}{\pP\pp{\PP}\parN\pP\q{\Q}\parN\pP\pr{\R}\parN\emptyset}
\endprooftree
$}

\bigskip

Our type system enjoys the properties of Session Fidelity and Subject
Reduction.  Moreover, it ensures the semantic property of
Lock-freedom.  Since every participant can freely perform outputs, to
prove this property we only have to show that all inputs can be
enabled.  For lack of space we only give 
the most interesting case in
the proof of Subject Reduction.
\begin{theorem}[Session Fidelity]\label{thm:sf}
If $\tyn{\G}{\Nt\parN\Msg}$ and 
$\G\parG\Msg\stackred\co\G'\parN\Msg'$, then $\Nt\parN\Msg\stackred{\co}\Nt'\parN\Msg'$ and  $\tyn{\G'}{\Nt'\parN\Msg'}$. 
\end{theorem}
\begin{theorem}[Subject Reduction]\label{thm:sr}
If $\tyn{\G}{\Nt\parN\Msg}$ and $\Nt\parN\Msg\stackred\beta\Nt'\parN\Msg'$, 
then $\G\parG\Msg\stackred\beta\G'\parG\Msg'$ and $ \tyn{\G'}{\Nt'\parN\Msg'}$.  
\end{theorem}
\begin{proof} The proof is by induction on $d=\weight(\G,\pp)$ where
  $\pp=\play\beta$. Notice that
  $\Nt\parN\Msg\stackred\beta\Nt'\parN\Msg'$ implies
  $\pp\in\plays\Nt$, which together with $\tyn{\G}{\Nt\parN\Msg}$
  implies $\pp\in\plays\G$. Then $d> 0$.  Moreover $d$ is finite since
  $\G$ is bounded.\\Let $d>1$ and $\G=\agtI{\pr}{\q}i I{\la}{\G}$ with
  $\pr\ne\pp$. Since $\tyn{\G}\Nt\parN\Msg$ must be derived using Rule
  \rn{In}, we
  get:\\
  \centerline{
    $\Nt\equiv\pP\pr{\inp{\q}{h}{H}{\la}{\R}}\parN\Nt_0$\qquad
    $\tyn{\G_{ j }}{\pP\pr{\R_{ j }}\parN\Nt_0\parN\Msg_{ j }}$ for
    all $j\in J$\qquad
    $ \n{\G_j}{\MM}$ for all  $j\in J$}\\
  where $J=\m{\set{\mq{\q_i}{\la_i}\pp}_{i\in
      I}}\Msg  =\m{\set{\mq{\q_h}{\la_h}\pp}_{h\in H}}\Msg$ and $
  \Msg\equiv\addMsg{\mq{\q_{j }}{\la_{ j }}\pr}{\Msg_{ j }}$ for all
  $j \in J$ and $\MM=\set{{\mq{\q_{ l }}{\la_{ l }}\pr}\mid l \in (I
    \cup H) \setminus J \ \&\ \q_l\neq\q_j\ \forall j\in J}$.  The
  condition $\pr\ne\pp$ ensures that the transition
  $\Nt\parN\Msg\stackred\beta\Nt'\parN\Msg'$ does not modify the
  process of participant $\pr$ and does not dequeue any message with
  receiver $\pr$ from $\Msg$.  Therefore we get
  $\Nt'\equiv\pP\pr{\inp{\q}{h}{H}{\la}{\R}}\parN\Nt'_0$ and $
  \Msg'\equiv\addMsg{\mq{\q_{j }}{\la_{ j }}\pr}{\Msg'_{ j }}$ for all
  $j \in J$. Moreover the transition can be done also if the process
  $\inp{\q}{h}{H}{\la}{\R}$ is replaced by an arbitrary process and
  top messages with receiver $\pr$ are dequeued. Therefore \\
 \centerline{
  $ \pP\pr{\R_j}\parN\Nt_0\parN\Msg_j\stackred\beta
    \pP\pr{\R_j}\parN\Nt_0'\parN\Msg'_j\text{ for all }j\in J $} 
    It is
  easy to verify that $\weight(\G_j,\pp)< \weight(\G,\pp)$. Then by
  induction we get $\G_j\parG\Msg_j\stackred\beta\G_j'\parG\Msg_j'$
  and $ \tyn{\G'_j}{\pP\pr{\R_j}\parN\Nt'_0}\parG\Msg_j'$ for all
  $j\in J$.  Let  $\G'=\agtI \pr\q i {I} \la {\G'}$ where $\G'_l =
  \G_{j_0}$ for some $j_0\in J$ and all $l\in I\backslash J$.   From
  $\G_j\parG\Msg_j\stackred\beta\G_j'\parG\Msg_j'$ we get
  $\G_j\parG\Msg\stackred\beta\G_j'\parG\Msg'$ for all $j \in J$. Then
  we derive $\G\parG \Msg \stackred\asCom \G'\parG \Msg' $ by Rule
  \rn{Inside-In}.  The condition $ \n{\G_j}{\MM}$ for all $j\in J$
  ensures that $\Msg'$ cannot contain a message $\mq{\q_k}{\la_k}\pr$
  with $k\in (I\cup H)\setminus J$ and $\q_k\neq\q_j$ for all $j\in
  J$.  The condition $\pr\neq\pp$ ensures that the transition
  $\stackred\beta$ cannot dequeue a message with $\pr$ as
  receiver. Hence $\m{\set{\mq{\q_h}{\la_h}\pp}_{h\in
      H}}{\Msg'} =J$.
  It is easy to verify that $ \n{ \G_j}{\MM}$ implies $
  \n{\G'_j}{\MM}$ for all $j\in J$.  From $
  \tyn{\G'_j}{\pP\pr{\R_j}\parN\Nt'_0}\parG\Msg_j'$ for all $ j \in
  J$ we get 
  $\plays{\G'_{ i }}\setminus\set\pr=\plays{\Nt'_0}$ for all $ i \in
  I$.  
Then  all premises of Rule \rn{In} hold  
and we can derive $\tyn{\G'}{\Nt'\parN\Msg'}$.
\end{proof}
\begin{theorem}[Lock-freedom]\label{thm:lf}
If $ \tyn{\G}\Nt\parN\Msg$  
and $\pP\pp\PP\in\Nt$,  then $\Nt\parN\Msg\stackred{\comseqA\cdot\co}$ with $\play\co=\pp$ for some $\comseqA$, $\co$.  
\end{theorem}
As expected, since queues in type configurations can arbitrarily grow, our type system is undecidable. 
\begin{theorem}[Undecidability]\label{thm:und}
Typing is undecidable. 
\end{theorem}

In order to recover from  this undecidability result, 
we can define an inductive version of typing, thus obtaining a sound algorithm. 
This inductive definition 
follows the  standard pattern  to deal with regular structures  for global types, and it requires the same queue at the beginning and at the end of each cycle.

\section{Related Work and Conclusion}\label{sec:relw}

We proposed flexible choice operators for an asynchronous 
multiparty  session calculus, in order to ensure the classical
session correctness properties for a larger class of protocols than is
usually done. Several other proposals for relaxing the constraints of
the original choice operator of~\cite{CHY08} were already mentioned
in~\refToSection{intro}.  We now discuss some of them in more detail.

In~\cite{CDG18}, which builds on~\cite{HY17}, we pushed this
flexibility even further by allowing 
input choices with different senders in processes,
without restrictions. 
The same approach was followed in~\cite{CDGH19}.
However, this liberal approach turned out to be incorrect, as pointed out
in~\cite{GHH21}, as it allows the following (synchronous) network to
be typed, while it is not deadlock-free. Indeed, this network can reach a
deadlock if $\pp$ chooses its second branch, leading both $\ps$ and
$\pt$ to choose their second branch too. Then, if $\pr$ chooses its
first branch, it will be unable to complete it.\\
\centerline{\begin{math}
\begin{array}{c}
\pP\pp{(\ps!a;\pt!a;\pr!d)\oplus(\ps!b;\pt!b)} \parN
\pP\pr{(\ps?c;\pt?e;\pp?d)+(\pt?e;\ps?c)} \parN\\
\pP\ps{(\pp?a;\pr!c)+(\pp?b;\pr!c)} \parN
\pP\pt{(\pp?a;\pr!e)+(\pp?b;\pr!e)}
\end{array} 
\end{math}}
In fact, this session is not race-free according to the
\emph{race-freedom} condition proposed in~\cite{GHH21}.  Note that the
asynchronous session obtained by composing this network with the empty
queue is not typable in our type system. Indeed, its typability would
contradict Subject Reduction or Lock-freedom, since it has the
derivative $\pP\pr{\pp?d}\parN\emptyset$ which is stuck.  This session
cannot be typed in~\cite{JY20}  either,  since participant $\pr$ does not
satisfy the  required  well-formedness conditions. Also the type system
of~\cite{MMSZ21} rejects this session, the reason being that
participant $\pr$ can read the same message in more than one
 branch.  More precisely, participant $\pr$ can read the message $c$
from participant $\ps$ and the message $e$ from participant $\pt$ in
both branches. This control is realised by annotating projections with
the set of available messages. We take advantage of queues in type
configurations for a similar but less refined control, which uses the
predicate ensuring  that a global type is inactive for a given set
of messages.

As future work, we plan  to investigate  three different variations
of our typing. The first one would be a weakening of condition 
$\n\G\MM$ in Rule \rn{In}, 
taking into account the order of sent messages and the expected
inputs.  The second one would be a strengthening of our typing in
order to forbid orphan messages.  The last one would be a weakening of
our typing in order to allow optional participants in the branches of
choices, possibly using connecting communications as
in~\cite{HY17,CDG18}.

 Another direction that would be worth investigating is the
relationship  between our approach and input races for sessions
based on Classical Linear Logic, see~\cite{KMW20,QKB21}.

To make  our type system more efficient  
we  will  
  design 
two algorithms, taking inspiration from~\cite{DGD22}, one for inferring global types for networks and the other one for checking the correctness of global types for queues, allowing also cycles in which queues  
increase. 

\textbf{Acknowledgments} This paper  came into being 
thanks to Ross Horne, who pointed out to us the example
reported in the Conclusion. We are indebted to him for many
interesting discussions on this subject and for suggestions on a
previous version of this paper. We  also  
thank the anonymous referees for helpful comments.

\bibliographystyle{eptcs}
\bibliography{sesdoi}

\end{document}